\documentclass[a4paper,UKenglish,cleveref, autoref, thm-restate,authorcolumns]{lipics-v2019}
\usepackage{algpseudocode}
\usepackage{wrapfig}
\usepackage{graphicx,amsmath,amssymb,amsthm,enumerate, mathtools, thmtools, thm-restate, float}
\MakeRobust{\Call}
\usepackage{algorithm2e}

\newtheorem{problem}{Problem}
\newcommand{\tdagger}{\textsuperscript{\textdagger}}

\usepackage{microtype}
\usepackage{algpseudocode}
\usepackage{wrapfig}
\MakeRobust{\Call}
\usepackage{algorithm2e}

\DeclareMathOperator{\OPT}{OPT}

\title{On the Complexity of BWT-runs Minimization\\ via  
Alphabet Reordering} 

\titlerunning{On the Complexity of BWT-runs Minimization via  
Alphabet Reordering
}

\author{Jason Bentley}{Department of Mathematics, University of Central Florida, 
USA
}
{jason.bentley@ucf.edu}
{}
{}

\author{Daniel Gibney}{Department of Computer Science, University of Central Florida, 
USA 
\url{https://www.cs.ucf.edu/~dgibney/} }{daniel.j.gibney@gmail.com}
{}
{}

\author{Sharma V. Thankachan}
{Department of Computer Science, University of Central Florida, 
USA 
\url{http://www.cs.ucf.edu/~sharma/}
}{sharma.thankachan@ucf.edu}
{}
{}

\authorrunning{J.\,B. Bentley, D. Gibney, and S.\,V. Thankachan} 

\Copyright{John Q. Public and Joan R. Public} 



\category{} 

\relatedversion{} 

\supplement{}

\funding{Supported in part by the U.S. National Science Foundation (NSF) under  CCF-1703489. }

\acknowledgements{We would like to thank Chandra Chekuri and Gonzalo Navarro for their helpful correspondence. 
 }

\nolinenumbers 

\EventEditors{John Q. Open and Joan R. Access}
\EventNoEds{2}
\EventLongTitle{42nd Conference on Very Important Topics (CVIT 2016)}
\EventShortTitle{CVIT 2016}
\EventAcronym{CVIT}
\EventYear{2016}
\EventDate{December 24--27, 2016}
\EventLocation{Little Whinging, United Kingdom}
\EventLogo{}
\SeriesVolume{42}
\ArticleNo{23}
\begin{document}
\maketitle
\begin{abstract}
The Burrows-Wheeler Transform (BWT) has been an essential tool 
in text compression and indexing. 
First introduced in 1994, it went on to provide the backbone for the first encoding of the classic suffix tree data structure in space close to entropy based lower bound. Within the last decade it has seen its role further enhanced with the development of compact suffix trees in space proportional to ``$r$'', the number of runs in the BWT. 
While $r$ would superficially appear to be only a measure of space complexity, it is actually appearing increasingly often in the time complexity of new algorithms as well. 
This makes having the smallest value of $r$ of growing importance.
Interestingly, unlike other popular measures of compression, the parameter $r$ is sensitive to the lexicographic ordering given to the text's alphabet.
Despite several past attempts to exploit this fact, a provably efficient algorithm for finding, or approximating, an alphabet ordering which minimizes $r$ has been  open for years. 

We help to explain this lack of progress
by presenting the first set of results on the computational complexity of minimizing BWT-runs via alphabet reordering.
We prove that the decision version of this problem is
NP-complete and cannot be solved in time $2^{o(\sigma + \sqrt{n})}$ unless the Exponential Time Hypothesis fails, where $\sigma$ is the size of the alphabet and $n$ is the length of the text. Moreover, we show that optimization variations of this problem yield strong inapproximability results.
Specifically, we show that the optimization problem is APX-hard.
In doing so, we relate two previously disparate topics: the optimal traveling sales person path of a graph and the number of runs in the BWT of a text. This provides a surprising connection between problems on graphs and text compression. 
In addition, by relating recent results in the field of dictionary compression, we illustrate that an arbitrary alphabet ordering provides an $O(\log^2 n)$-approximation. 

Lastly, we provide an optimal linear time algorithm for a more restricted problem of finding an optimal ordering on a subset of symbols (occurring only once) under ordering constraints.
We also look at a generalization of
this problem on a newly discovered class of graphs with BWT like properties called Wheeler graphs where we show its NP-completeness. 
\end{abstract}
\newpage
\section{Introduction and Related Work}

The Burrows-Wheeler Transform (BWT) is an essential building block in the fields of text compression and indexing with a myriad of applications in bioinformatics and information retrieval~\cite{langmead2009ultrafast,li2010fast,li2009soap2,navarro2016compact}.
Since it first arose in 1994~\cite{burrows1994block}, it has been utilized to provide the popular compression
algorithm bzip2 and has been adapted to provide powerful compressed text indexing data  structures, such as the FM-index~\cite{DBLP:conf/focs/FerraginaM00}.
Hence, improvements to the algorithmic aspects of this transformation and related data structures can  have a significant impact on the research community.

\begin{wrapfigure}[]{r}{0.32\textwidth}
    \includegraphics[width=\linewidth]{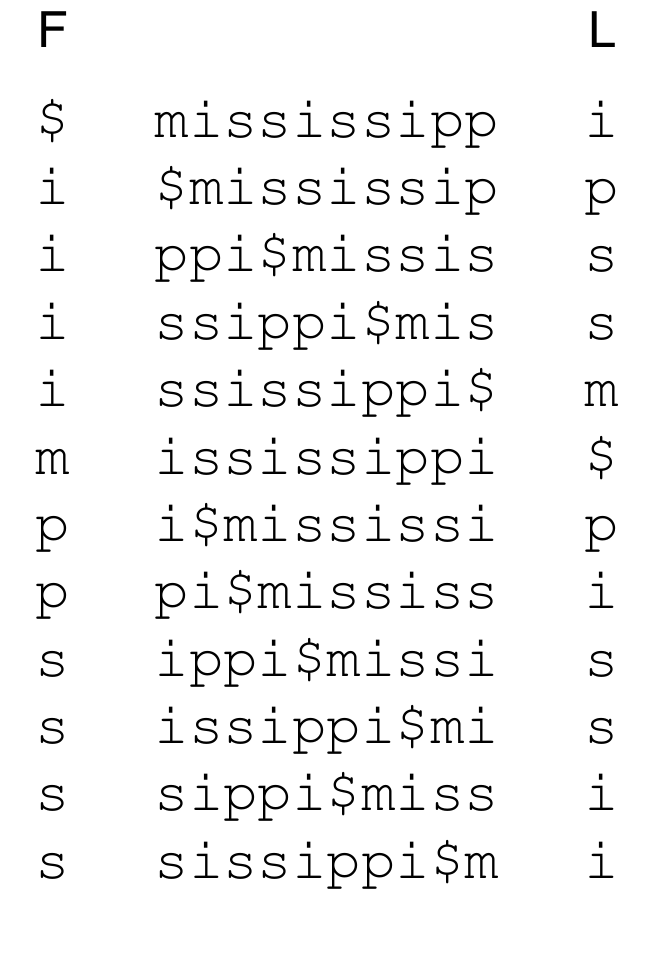}
    \caption{Column L ($r=9$) shows the 
    BWT of $mississippi$. 
    }
    \label{fig:bwt_example}
\end{wrapfigure}
The BWT of a text $T[1,n]$, denoted by $BWT(T)$ is a reversible transformation which can be defined as follows: sort the circular shifts of $T$ in lexicographical order and place the sorted circular shifts in a matrix. By reading the last column of this matrix from top to bottom we obtain $BWT(T)$.  To make the transformation invertible a new symbol $\$$ (lexicographically smaller than others) is appended to $T$ prior to sorting the circular shifts.
See Figure \ref{fig:bwt_example}
for an example.
Historically, the BWT was introduced 
for the purpose of text compression~\cite{burrows1994block}, where it's effectiveness is based on symbols with shared preceding context forming long \emph{runs} (maximal unary substrings).


Recently, the number of runs ``$r$'' 
in the BWT 
has become of increasing interest.
This can be attributed to the fact that 
many modern text collections are highly repetitive, which makes their 
 compression effective via the BWT followed by Run-Length encoding (i.e., in space proportional to $r$). 
This raised an interesting question: can we also index the text in space propositional to $r$? Note that the FM-index needs space proportional to $n$ (i.e., $\approx n\log\sigma$ bits, where $\sigma$ is the alphabet size).
The data-structure community has made great strides in answering this question~\cite{bannai2018online,boucher2019prefix,DBLP:conf/soda/GagieNP18,KempaP18,KuhnleMBGLM19,ohno2018faster}.
The first such index was developed 
by M{\"a}kinen and Navarro 
in 2005 \cite{DBLP:conf/cpm/MakinenN05}. However, it lacked the ability to efficiently locate the occurrences of a pattern within space $\tilde{O}(r)$.
After a decade of related research~\cite{MakinenNSV09,DBLP:conf/soda/GagieNP18},
we now have fully
functional suffix trees in space proportional to $r$,
developed by Gagie {\it et al.}~\cite{JACM2020}. 
Also note that the
recent optimal BWT construction algorithm for highly repetitive texts is parameterized by $r$~\cite{DBLP:conf/soda/Kempa19}.
A technique reducing the value of this parameter $r$ would  have a significant impact on a large body of work.

A natural way to minimize $r$ is to change the lexicographic ordering assigned to symbols of the alphabet. 
To demonstrate that this can have an impact on $r$, consider as an example the text $mississippi$ with the usual ordering $\$ < i < m < p < s$ where $r = 9$, but with the ordering $\$ < s < i < p < m$ we have $r = 8$. 
In fact, there exist string families in which $r$ differs by a factor of $\Omega(\log n)$ for different orderings.  
This problem of reordering the alphabet is clearly fixed-parameter tractable in alphabet size $\sigma$ and has a trivial $O(\sigma!~ n)$ time solution. 
This may be adequate  when $\sigma$ is small as in DNA sequences. However, this is far from satisfactory from a theoretical point of view, or even  from a practical point
when the alphabet is  slightly larger, such as in protein sequences, natural language texts, ascii texts, etc.

\textbf{Related Work:} A related work in 2018  on block sorting based transformations by Giancarlo {\it et al.} gives a theoretical treatment of alphabet ordering in the context of the Generalized BWT~\cite{DBLP:conf/dlt/GiancarloMRRS18}. It was shown that for any alphabet ordering, $r$ is at most twice the number of runs in the original text, a result which then holds for the standard BWT as well. Note however that this gives no lower bound on $r$ and thus gives no results on the approximability of the run minimization problem. 
There have been multiple previous attempts to develop other approaches to alphabet ordering. In the context of bioinformatics the role of ordering on proteins was considered in \cite{yang2010use} with approaches evaluated experimentally. 
Similar heuristic approaches evaluated through experiments were done in~\cite{DBLP:journals/spe/Abel10}. Researchers have also consider more restricted versions of this problem. For example, one can try to order a restricted subset of the alphabet, or limit where in the ordering symbols can be placed. On this problem heuristics have been utilized. 
In the field of bioinformatics, software tools like BEETL utilize these techniques to handle collections of billions of reads~\cite{BEETL}. 
Another related work in~\cite{CazauxR19} shows, how to permute a given set of strings in linear time, such that the number of runs in the BWT of the (long) string
obtained by concatenating the input strings, separated by the same delimiter symbol is  minimized. 

Given the lack of success with attacking the main problem from the upper bound side, perhaps it is best to approach the problem from the perspective of lower bounds and hardness. 
To this end, we show why a provably efficient algorithm has been evasive.


\section{Problem Definitions and Our Results}
For the following
problems we consider all strings to be over an alphabet $\Sigma$ of size $\sigma$.
A run in a string $T$ is a maximal unary sub-string. Let $\rho(T)$ be the number of runs in $T$.

\begin{problem}[Alphabet Ordering (AO)] \label{pro:AO}
Given a string $T[1,n]$ 
and an integer $t$, decide whether there exists an ordering of the symbols in its alphabet 
such that $\rho(BWT(T)) \leq t$.
\end{problem}

\begin{theorem}
\label{thm:AO_NPC}
AO  is NP-complete and its corresponding minimization problem is APX-hard. 
\end{theorem}
The problem can be solved in $n \cdot \sigma! = n\cdot 2^{O(\sigma\log\sigma)}$ time naively. However, any significant improvement seems unlikely as per the Exponential Time Hypothesis (ETH)~\cite{DBLP:journals/eatcs/LokshtanovMS11}. 
\begin{corollary}
\label{cor:AO_exp_time}
Under ETH, AO cannot be solved in time $2^{o(\sigma + \sqrt{n})}$.
\end{corollary}

It is known that 
$\rho(BWT(T))$
can be lower bounded by the size of string attractor $\gamma$, a recently proposed compressibility measure~\cite{DBLP:conf/ictcs/Prezza18}. Kempa and Kociumaka showed that
$\rho(BWT(T))$
can be upper bounded by  $O(\gamma\log^2 n)$~\cite{DBLP:journals/corr/abs-1910-10631}. However, $\gamma$ is independent of the alphabet ordering
and the following result is immediate.

\begin{corollary}
\label{thm:AO_approx}
Any alphabet ordering is an $O(\log^2 n)$-approximation for AO.
\end{corollary}

We now introduce two variants of AO: (i) a specialization of AO  where we impose more constraints on the ordering, suitable for the BWT of a set of strings and (ii) 
 a generalization to the class of graphs known as Wheeler graphs which allow for BWT based indexing~\cite{DBLP:journals/tcs/GagieMS17}.

\begin{problem}[Constrained Alphabet Ordering (CAO)]\label{prob:delimiter_ordering}
Given a set of $d$ strings $T_0, \hdots, T_{d-1}$ of total length $N$, 
find an ordering $\pi$ on the symbols $\$_i$ $(0\leq i \leq d-1)$ such that $\$_{\pi(0)} \prec \$_{\pi(1)} \hdots \prec \$_{\pi(d-1)} \prec 0 \hdots \prec \sigma-1$ and $\rho(BWT(T_0\$_0T_1\$_1 \hdots T_{d-1}\$_{d-1}))$ is minimized.
\end{problem}

We call 
$\$_0,\$_1,\dots,\$_{d-1}$
\emph{special symbols}. 
In Section~\ref{sec:DO}, we provide an example where an optimal ordering of special symbols removes a factor of $\Omega(\log_\sigma d)$ in the number of runs, demonstrating that this can be a worthwhile preprocessing step.
We refer to~\cite{BEETL} for an immediate use case
in bioinformatics, where the input is a large collection of
DNA reads.

\begin{theorem}
\label{thm:delim_odering}
CAO can be solved in optimal $O(N)$ time,  where 
$N = |T_0|+|T_1+\dots+|T_{d-1}|$.
\end{theorem}


As an extension of CAO we consider the Source Ordering Problem (SO) on Wheeler graphs. Here the goal is to order the source vertices of a Wheeler graph $G$ in order to minimize the number of runs in the induced string denoted by $BWT(G)$ (exact definition of $BWT(G)$ is deferred to Section \ref{sec:AO} and Appendix \ref{appendix:wheeler}). In contrast to  
CAO, 
SO 
is computationally difficult.

\begin{problem}[Source Ordering(SO)]
Given a Wheeler graph $G$ and an integer $t$, decide whether there exists an ordering of the sources such that $\rho(BWT(G)) \leq  t$.
\end{problem}

\begin{theorem}
\label{thm:SO_NPC}
SO is NP-complete and corresponding minimization problem is APX-hard. 
\end{theorem}


\section{Preliminaries}
\label{sec:preliminaries}

\begin{wrapfigure}[]{r}{0.32\textwidth}
\vspace{-20mm}
    \includegraphics[width=\linewidth]{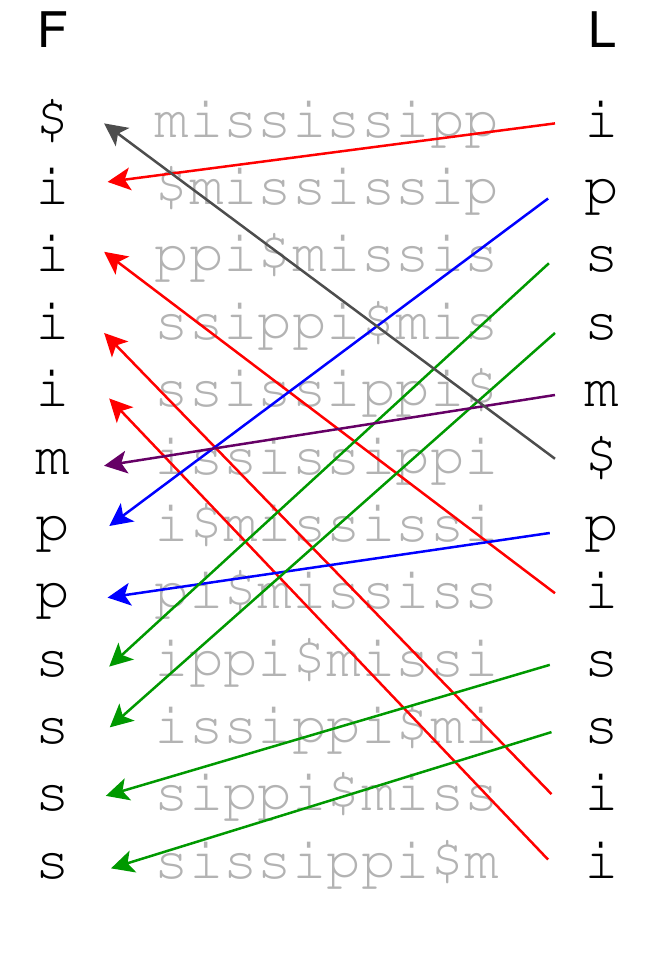}
    \caption{The LF-mapping between the last and first columns of sorted circular shifts matrix. 
    Also, note that LF-mappings from the same symbol (shown in the same color) do not cross.
    }
    \label{fig:LF_map}
    \vspace{-2mm}
\end{wrapfigure}

\subsection{The LF-map}
The key observation which makes BWT based indexing possible is the LF-mapping. The LF-mapping maps a symbol in the column $L$ to its corresponding position in the column $F$, where $L$ and $F$ are the last and first columns respectively of the matrix of sorted circular shifts. 
One observation we make now
is that if we were to write out the directed path obtained by following the LF-mapping from row to row, creating a vertex for each row, and marking each vertex with the symbol seen in the $L$ column, we would get a directed path with vertices labeled with the reverse of the string $T$. For example, in Figure \ref{fig:bwt_example} if we start from the first row, we obtain the path $i \rightarrow p \rightarrow p \rightarrow i \rightarrow s \rightarrow s \rightarrow i \rightarrow s \rightarrow s \rightarrow i \rightarrow m \rightarrow \$$. This view of the LF-mapping will be used throughout,
as all LF-mappings will be illustrated as directed paths labeled in this fashion. The second observation we make is that the relative order of identical symbols is preserved by the LF-mapping, i.e. in Figure \ref{fig:bwt_example} the arrows of the same color do not cross. Our hardness results will use this property extensively.

\subsection{L-reductions}
Our inapproximability results 
use L-reductions~\cite{crescenzi1997short}.
We will use the following notation: 
 \begin{itemize}
\item $\OPT_A(x)$ denotes the cost of an optimal solution to the instance $x$ of Problem $A$. 
\item $c_A(y)$ denotes the cost of a solution $y$ to an instance $x$ of Problem $A$ (suppressing the $x$ in the notation $c_A(x,y)$).
\item Since all problems presented here are minimization problems the approximation ratio can be written as $R_A(x,y) =  \frac{c_A(y)}{\OPT_A(x)}$, which is $\geq 1$.
\item Let $f_A(x) = x'$ be a mapping of an instance $x$ of Problem $A$ to instance $x'$ of Problem $B$. The mapping $f_A$ must be computable in polynomial time.
\item Let $y'$ be a solution to instance $x' = f_A(x)$, and $g_B(y') = y$ the mapping of a solution $y'$ to a solution $y$ for instance $x$.  The mapping $g_B$ must be computable in polynomial time.
\end{itemize}
%
%
Taking $x$, $y$, $x'$ $y'$ as above, an L-reduction is defined by the pair of functions $(f_A, g_B)$ such that 
there exist constants $\alpha, \beta > 0$ where for all $x$ and $y$ the following conditions hold:
$$~~\text{(i)}~~\OPT_B(f_A(x)) \leq \alpha \OPT_A(x)~~~~\text{and}~~~~  \text{(ii)}~~c_A(g_B(y')) - \OPT_A(x) \leq \beta\Big(c_B(y') - \OPT_B(f_A(x))\Big).~~$$
As a result, $R_B(x', y') = 1 + \varepsilon$ implies $R_A(x,y) \leq 1 + \alpha \beta \varepsilon = 1+O(\varepsilon)$. \\
L-reductions preserve APX-hardness~\cite{DBLP:journals/jcss/PapadimitriouY91}. 


\section{Hardness of Alphabet Ordering}
\label{sec:AO}
We will demonstrate a sequence of L-reductions from the
\textbf{(1,2)-TSP Path Problem} where the aim is to find a Hamiltonian Path of minimum weight through an undirected complete graph on $n$ vertices where all edges have weights either 1 or 2. 

\begin{lemma}
[\cite{DBLP:journals/mor/PapadimitriouY93}]
\label{1,2-TSP-lemma}
(1,2)-TSP Path is APX-hard, even with only $\Theta(n)$ edges of weight 1.
\end{lemma}

\subsection{Reduction Phase 1}
\begin{wrapfigure}[]{r}{0.33\textwidth}
    \includegraphics[width=\linewidth]{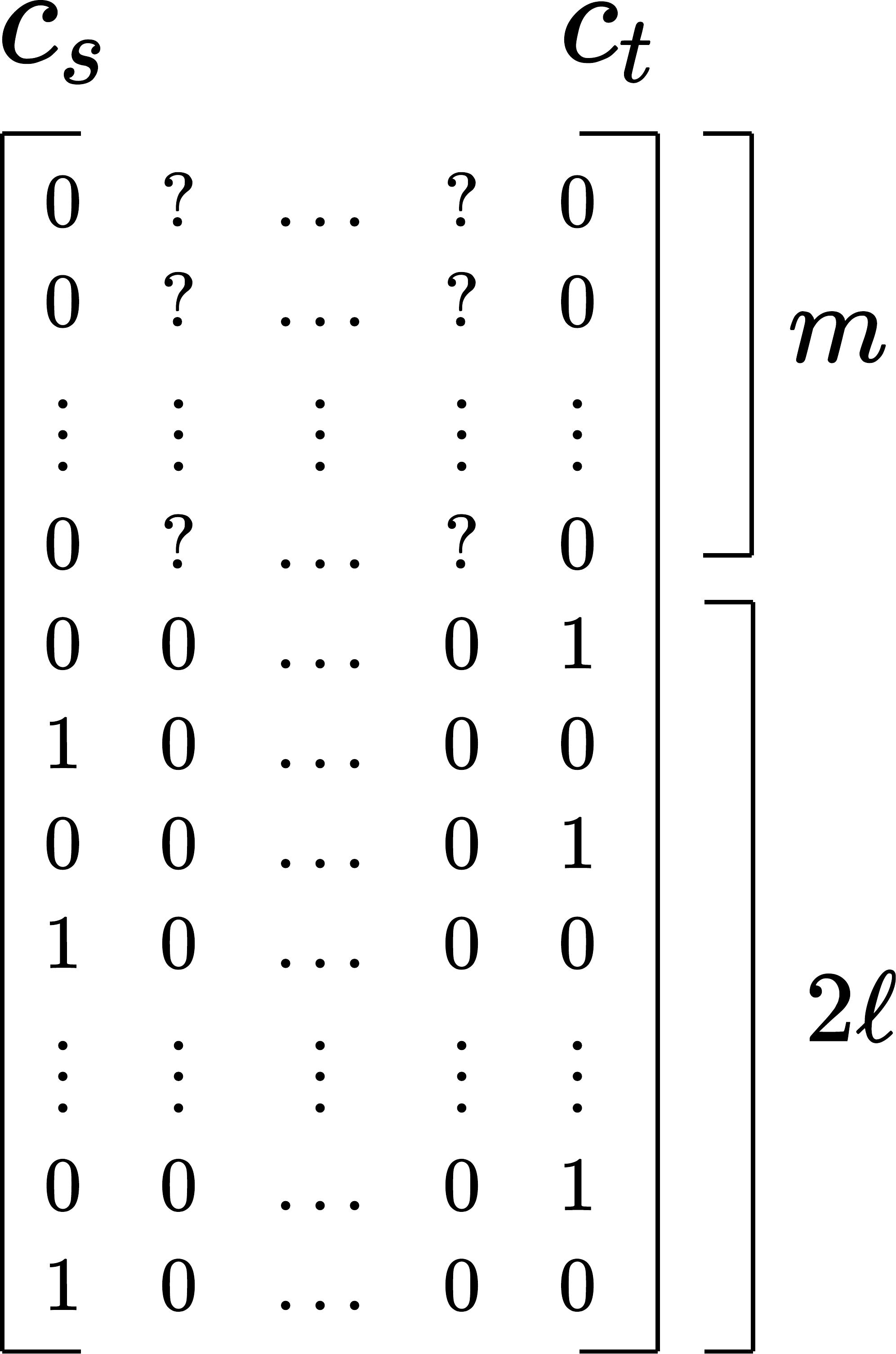}
    \caption{The modified incidence matrix for the graph $G$. Each of the first $m$ rows is for an edge. The bottom $2\ell = 8m$ rows are added as are the outer two most columns.
    }
    \label{fig:matrix_M}
\end{wrapfigure}
Given a graph with $m = \Theta(n)$ edges of weight 1 as input to (1,2)-TSP Path, remove all edges of weight 2. We call the resulting graph $G$.
Construct the incidence matrix for $G$ (a row for each edge, and column for each vertex, where the two 1's in a row indicate which two vertices are incident to the edge for that row). Then add $2\ell$ rows of all 0's to bottom of the matrix, where $\ell$ will be fixed later.
Next, add two additional columns $c_s$ and $c_t$ 
where $c_s[i] = 1$ if $i \in \{ m+2, m+4, \hdots, m+2\ell\}$ and $0$ otherwise, and $c_t[i] = 1$ if $i \in \{m+1, m+3, \hdots, m+2\ell-1\}$ and $0$ otherwise (see Figure \ref{fig:matrix_M}). The final matrix we denote as $M$. We will call this intermediate problem Column Ordering (CO), which is: given a matrix $M$ constructed as above, find an optimal ordering on the columns so as to minimize the number of runs in its linearization. Here the linearization of $M$ under column ordering $\pi$, denoted $L(M_\pi)$, is the string obtained by ordering $M$'s columns by $\pi$ and then concatenating the rows of $M$ from top to bottom.

Ignoring the added columns $c_s$ and $c_t$, the ordering $\pi$ induces a collection of disjoint paths in $G$, which we call $P$, where two vertices form an edge if their columns are adjacent and there exists a row with 1's in both columns.
Given $P$ we create a (1,2)-TSP Path by connecting the paths in $P$ with $|P|-1$ edges of weight 2.  

\begin{lemma}
\label{lem:CO_costs}
If $c_s$ and $c_t$ are the first and last columns of $M$ respectively, then the cost of our CO solution is $\rho(L(M_\pi)) = 2m_1 + 4(m-m_1) + 2\ell + 1$, where $m_1$ is the number of rows whose edges are in the collection of paths $P$. The corresponding cost of the solution to (1,2)-TSP Path is $m_1 + 2(n-1-m_1)$. 
\end{lemma}
\begin{proof}
Ignoring the starting character of $L(M_\pi)$ for the moment,
every row in $M$ corresponding to an edge in $P$ contributes two runs to $\rho (L(M_\pi))$ (e.g. $0 \hdots 0110 \hdots 0$). Any row whose edge is not in $P$ and not in the bottom $2\ell$ rows, contributes four (e.g. $0\hdots 010 \hdots 010 \hdots 0$) and there are $m-m_1$ such (rows) edges.  The extra $2\ell$ rows all contribute one run. Adding the `$+ 1$' term for the start of $L(M_\pi)$ gives the desired expression.
The second statement follows from the TSP path having $m_1$ edges of weight 1 and the overall number of edges needed to form a Hamiltonian path.
\end{proof}

\noindent
We now fix the value $\ell = 4m$.
\begin{lemma}
\label{lem:non_ends_sub_opt}
If $c_s$ and $c_t$ are not the first and last columns respectively, then the solution to CO is sub-optimal.
\end{lemma}
\begin{proof}
If $c_t$ is first and $c_s$ is last, then one extra run in contributed over $c_s$ being first and $c_t$ last, while maintain the rest of the ordering to be the same. In any configuration where either $c_s$ or $c_t$ are not ends of the matrix the bottom rows will contribute at least $3\ell$.
Letting $m_1^*$ denote the optimal number of edges of $P$, then
\begin{align*}
2m_1^* + 4(m-m_1^*) + 2\ell + 1 < 4m + 2\ell \leq 3\ell.
\end{align*}
Note that the first inequality is strict since we can always find at least one edge for $P$.
\end{proof}
It is immediate from Lemmas \ref{lem:CO_costs} and \ref{lem:non_ends_sub_opt} that an optimal solution for CO is one which maximizes $m_1$, and this provides an optimal solution for (1,2)-TSP Path. We now must show that our reduction is also an L-reduction. Lemmas \ref{lem:L-red_first_and_last} and \ref{lem:L-red_not_first_and_last} consider the two possible cases.

\begin{lemma}
\label{lem:L-red_first_and_last}
If $c_s$ and $c_t$ are the first and last columns respectively in a solution to CO, then the L-reduction conditions hold.
\end{lemma}
\begin{proof}
By Lemmas \ref{lem:CO_costs} and \ref{lem:non_ends_sub_opt}, the optimal cost for the instance of CO can be expressed as $2m_1^* + 4(m-m_1^*) + 2\ell + 1$ and the optimal cost for the instance of (1,2)-TSP Path as $m_1^* + 2(n-1-m_1^*)$. To prove Condition (i), we need to show there exists an $\alpha > 0$ such that
\begin{align*}
2m_1^* + 4(m-m_1^*) + 2\ell + 1 \leq \alpha (m_1^* + 2(n-1-m_1^*)).
\end{align*}
Since $m = \Theta(n)$ there exists a constant $C$ such that for $n$ large enough $m \leq Cn$ and the left hand side can be bound above by (recall $\ell = 4m$)
\begin{align*}
    2(m_1^* + 4(Cn-m_1^*)) + 8Cn + 1 = 16Cn - 6m_1^* + 1
\end{align*}
and since $m_1^* \leq n-1$
it is easy to find such an $\alpha$ (for example $\alpha = 32C + 1$ when $n\geq 2$). For Condition (ii) we need a constant $\beta > 0$ such that
\begin{align*}
    m_1 + &2(n-1-m_1) - (m_1^* + 2(n-1-m_1^*))\\ \leq &\beta(2m_1 + 4(m-m_1) + 2\ell + 1-(2m_1^* + 4(m-m_1^*) + 2\ell + 1)).
\end{align*}
This can be rewritten as $m_1^* - m_1 \leq 2\beta(m_1^* - m_1)$, which is true for any $\beta \geq 1/2$.
\end{proof}

\begin{lemma}
\label{lem:L-red_not_first_and_last}
If $c_s$ and $c_t$ are not the first and last columns respectively in a solution to CO, the L-reduction conditions still hold.
\end{lemma}

\begin{proof}
Condition (i) holds since the optimal solution values to the overall problem have not changed. 
\noindent
For Condition (ii), we consider the two scenarios:
\begin{itemize}
    \item \textbf{Scenario 1}: $c_s$ or $c_t$ are not at the far ends of $M$.  Then the cost of the solution for CO, which is at least $3\ell$, exceeds the cost for any solution considered in Lemma \ref{lem:L-red_first_and_last}. Furthermore, any corresponding solution for (1,2)-TSP Path was already considered in Lemma \ref{lem:L-red_first_and_last}, where now the right-hand is larger than it was in Lemma \ref{lem:L-red_first_and_last}.
    \item \textbf{Scenario 2}: $c_t$ is the first column of $M$ and $c_s$ is the last. Then, again, we have already considered a solution in Lemma \ref{lem:L-red_first_and_last} which has solution cost one less for CO and yet had the same solution cost for (1,2)-TSP Path.
\end{itemize}
This completes the proof. 
\end{proof}

We now have our L-reduction from (1,2)-TSP to column ordering, the Phase 1 of our reduction. 

\subsection{Reduction Phase 2}
Given the matrix $M$ as constructed in Phase 1 from $G$, we will now construct a string $T$ as input to the problem AO. It is easier to describe $T$ in terms of its substrings, which are created by iterating through the matrix $M$ as follows:

\begin{itemize}
    \item  For $1\leq j \leq n$, $1 \leq i \leq m + 2\ell$: 
     if $M_{i,j} = 1$ output the substring $10^{i+1}2C_j$

    \item For $1\leq j \leq n$:
            output the substring  $0^{m+ 2\ell + 2}2C_j$
    \item Append to each substring created above a unique $\$_i$ symbol ($1 \leq i \leq 2m+2\ell+n$).
\end{itemize}

The string $T$ is the concatenation of these substrings in any order. We refer the reader to Figures \ref{fig:red_values} and \ref{fig:red_values_full} for illustrations which will be helpful in understanding future arguments.
The alphabet set $\Sigma$ is $\{0,1,2\}\cup \{C_1, C_2, \hdots, C_n\} \cup \{\$_1, \$_2, \hdots, \$_{2m + 2\ell + n}\}$ and its size $\sigma = 3+2n+2m + 2\ell$.

Given a solution $\pi$ to this instance of AO we use the relative ordering given to the $C_i$ symbols as the ordering for the columns of $M$.
For the analysis of why this works, we define some properties which we would like $BWT(T)$ and $\pi$ to have. For any symbol $a \in \Sigma$ we will call the maximal set of indices where the $F$ column of the sorted circular shift matrix has only $a$'s as the \emph{$a$-block}. We let $C_s$ and $C_t$ denote the symbols for columns $c_s$ and $c_t$ respectively.\\

\noindent
\textbf{Desired Properties:}
\begin{enumerate}
    \item For a fixed $j$, all $C_j$ symbols are placed adjacently in $BWT(T)$;
    \item All $2$ symbols are placed adjacently in $BWT(T)$;
    \item The symbol $2$ is immediately preceding the symbol $0$ in $\pi$;
    \item In $\pi$ the $\$_i$ symbols are ordered in such a way as to minimize the number of runs of 1 in the 0-block.
    \item The symbols $C_s$ and $C_t$ are both positioned at the beginning and end respectively of the alphabet ordering given to the $C_i$ symbols.
\end{enumerate}

\noindent
Let $r_0$ denote the number of runs created in the 0-block, minus the number of $\$_i$ symbols in the 0-block. 
\begin{lemma}
\label{lem:AO_no_desired_suboptimal}
Unless all of the desired properties hold, the solution to AO is suboptimal.
\end{lemma}
\begin{proof}
If any of Properties 1-3 are violated, we can exchange our solution with one which maintains the value $r_0$ but reduces the runs created in other blocks. This is since the alphabet ordering can be modified to have these properties, while at the same time maintaining the relative orderings of LF-mappings 
into the 0-block.
In the case of Property 4, given that Properties 1-3 hold, modifying the solution so that the property holds can only decrease $r_0$, while it maintains the number of runs created in other blocks.

Assuming properties 1-4 hold, there are two possibilities, either $C_s$ and $C_t$ are extremal or they are not. 
In the case of being extremal, by Property 4 the $2\ell = 8m$ instances of 1's in the bottom $2\ell$ rows of $M$ shall correspond to $4m$ runs of two consecutive 1's (if $C_s < C_t$) or $4m-1$ runs of two consecutive 1's and 2 runs of lonely 1's (if $C_t < C_s$) in the 0-block of $BWT(T)$ (by rearranging $\$_i$'s as needed). The upper rows of $M$ shall correspond to at most $2m$ runs of 1's in the 0-block of $BWT(T)$. Hence, in the 0-block there are at most $6m+1$ runs of 1's making at most $6m + 2$ runs of zeros to surround them, making $r_0 \le 12m + 3$.
In the case where $C_s$ and $C_t$ are not extremal, there are $8m$ runs of lonely 1's in the 0-block of $BWT(T)$, at least $8m+1$ runs of 0's to surround them, and a $\$_i$ symbol for each column, leading to $r_0 \ge 16m + 1$. 
It is evident to choose the first option which necessitates $C_s$ and $C_t$ to be extremal.
Choosing $C_t < C_s$ forms one more run in $BWT(T)$ with respect to the bottom rows of $M$ compared to $C_s < C_t$. However, the minimum number of runs in $BWT(T)$ corresponding to the upper rows of $M$ remain unaffected. 
\end{proof}
\begin{figure}
    \begin{minipage}{.5\textwidth}
    \centering
    \includegraphics[width=1\textwidth]{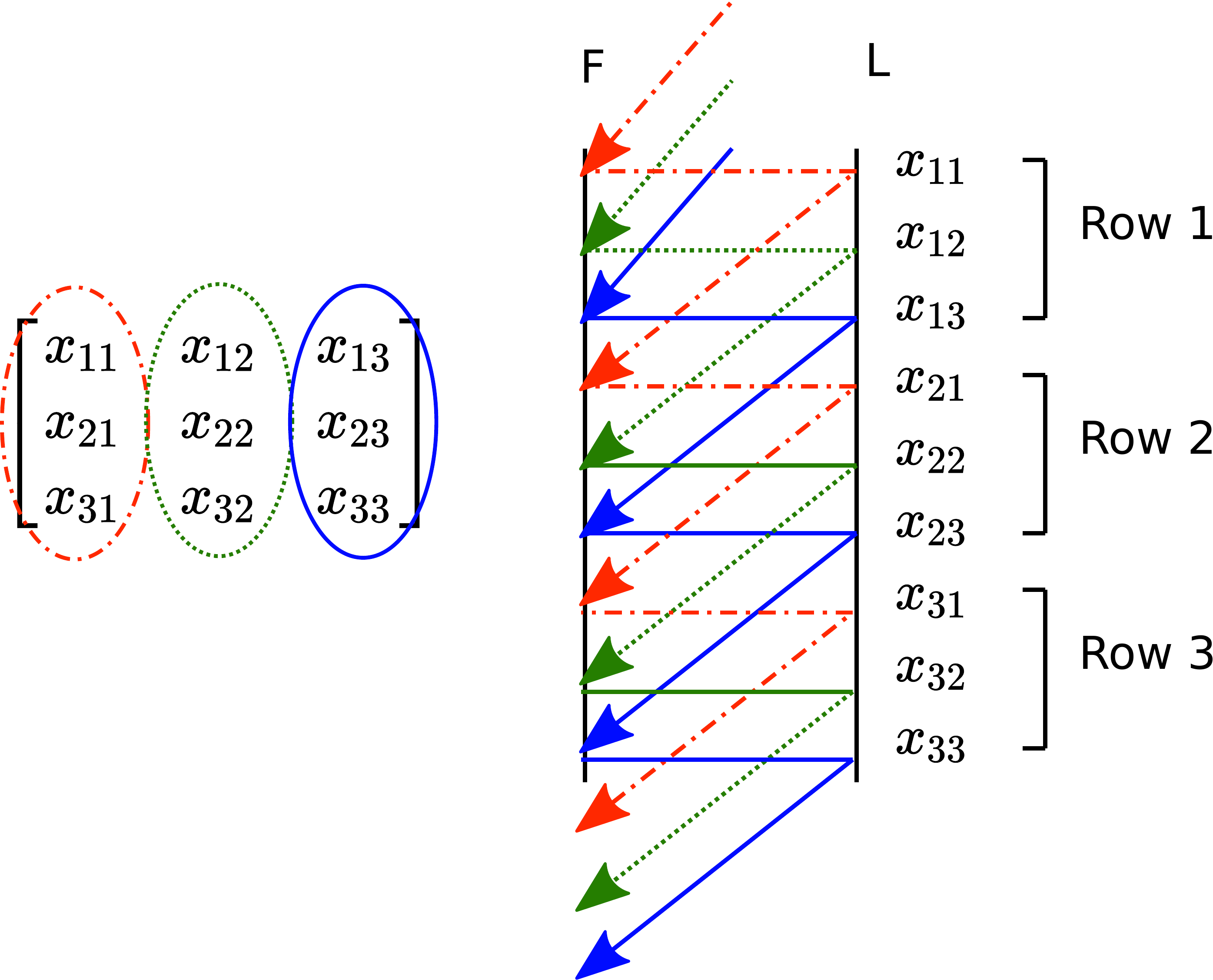}
    \end{minipage}
    \hspace{1em}%
    \begin{minipage}{.5\textwidth}
    \centering
    \includegraphics[width=1\textwidth]{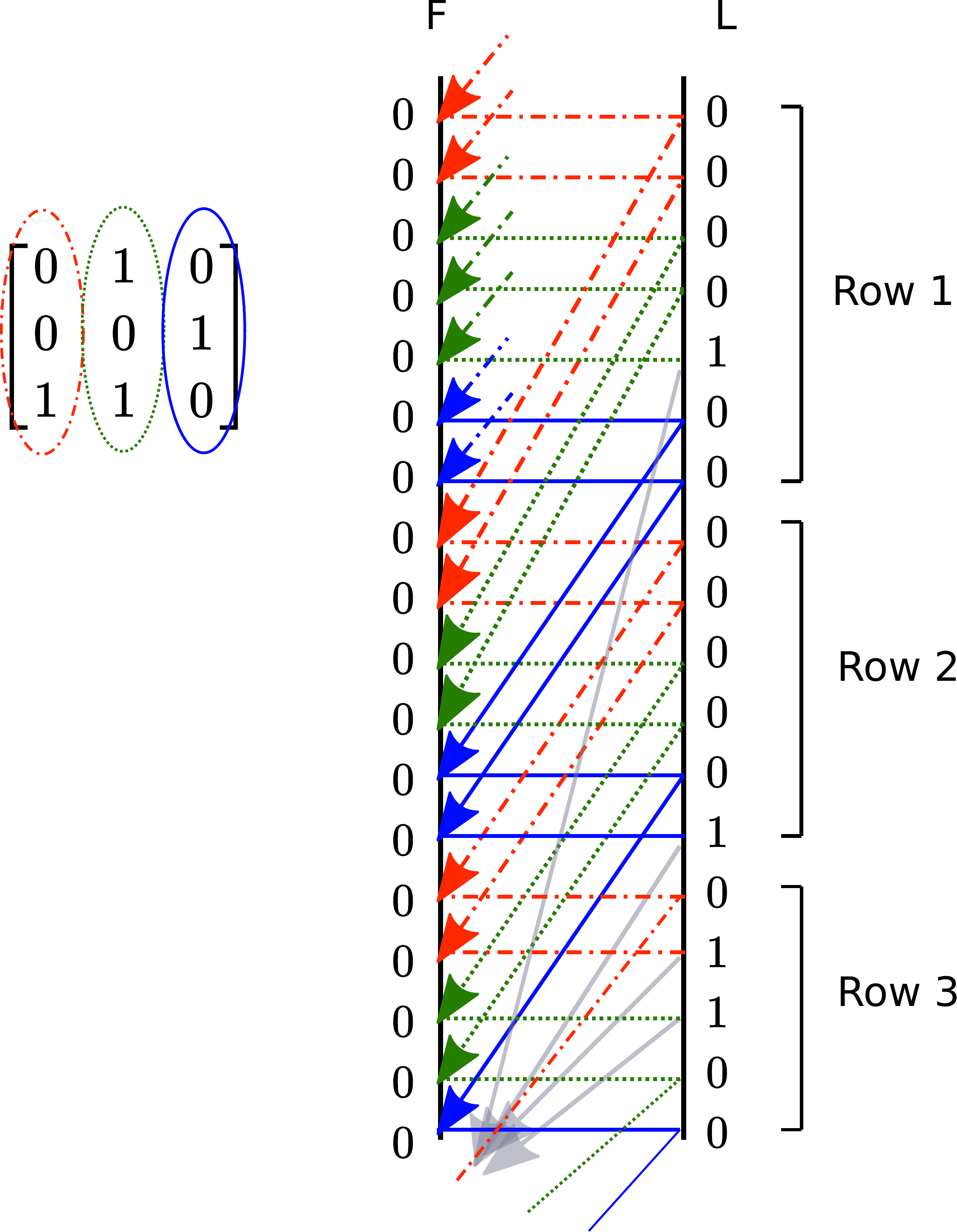}
    \end{minipage}
    \caption{On the left - a simplistic view of how we would like to simulate the linearization of $M$. On the right - the column bundling technique necessary to perform the simulation. Note that this matrix has a run of three 1's, a situation that will not arise if Properties 1-5 are met.}
    \label{fig:red_values}
\end{figure}
\begin{lemma}
\label{lem:desired_cost}
If all Properties 1-5 hold, then $r_0 = \rho(L(M_\pi)) - 1$ and $\rho(BWT(T)) = r_0 + \sigma - 1$.
\end{lemma}
\begin{proof}
We will first show that when Properties 1-5 hold $r_0 = \rho(L(M_\pi)) - 1$. Our aim will be to have a substring of $BWT(T)$ within the 0-block which is the same as $L(M_\pi)$ except for the lengths of its runs. We will call this substring the simulation of $L(M_\pi)$.
Ideally, we could have the situation as illustrated in the left-hand side of Figure \ref{fig:red_values}.  If each element in $M$ were 0 this would work, thanks to the LF-mappings coming from the same symbol are order preserving (not crossing). But, since there are 1's in $M$, there will be LF-mappings leaving the 0-block. We will show how to create the simulation of $L(M_\pi)$ next.

\textbf{`Column bundles':} We simulate $L(M_\pi)$ with a collection of directed paths which we will call a ``column bundle''.
A column bundle for column $j$ is the  collection of all directed paths (recall Section \ref{sec:preliminaries} where we viewed the LF-mapping as creating a directed path) created by the substrings of the form $\{0,1\}^*2C_j$. In addition, we shall call the directed path created by each individual substring of this form a ``thread''.
These $C_j$ symbols essentially bind these threads together. Using this technique we are now able to simulate $L(M_\pi)$ as in the right-hand side of Figure \ref{fig:red_values}. A collection of adjacent 0's in the same column bundle now represents $x_{ij} = 0$, and a 1 occurring in the simulation is represented by a single thread which leaves the 0-block by mapping into the 1-block, while the rest of its column's bundle continues in the simulation. In addition, for each column there exists a thread consisting of entirely 0's for the length of the simulation, preventing runs from being missed.
\begin{figure}[H]
    \centering
    \includegraphics[width=\textwidth]{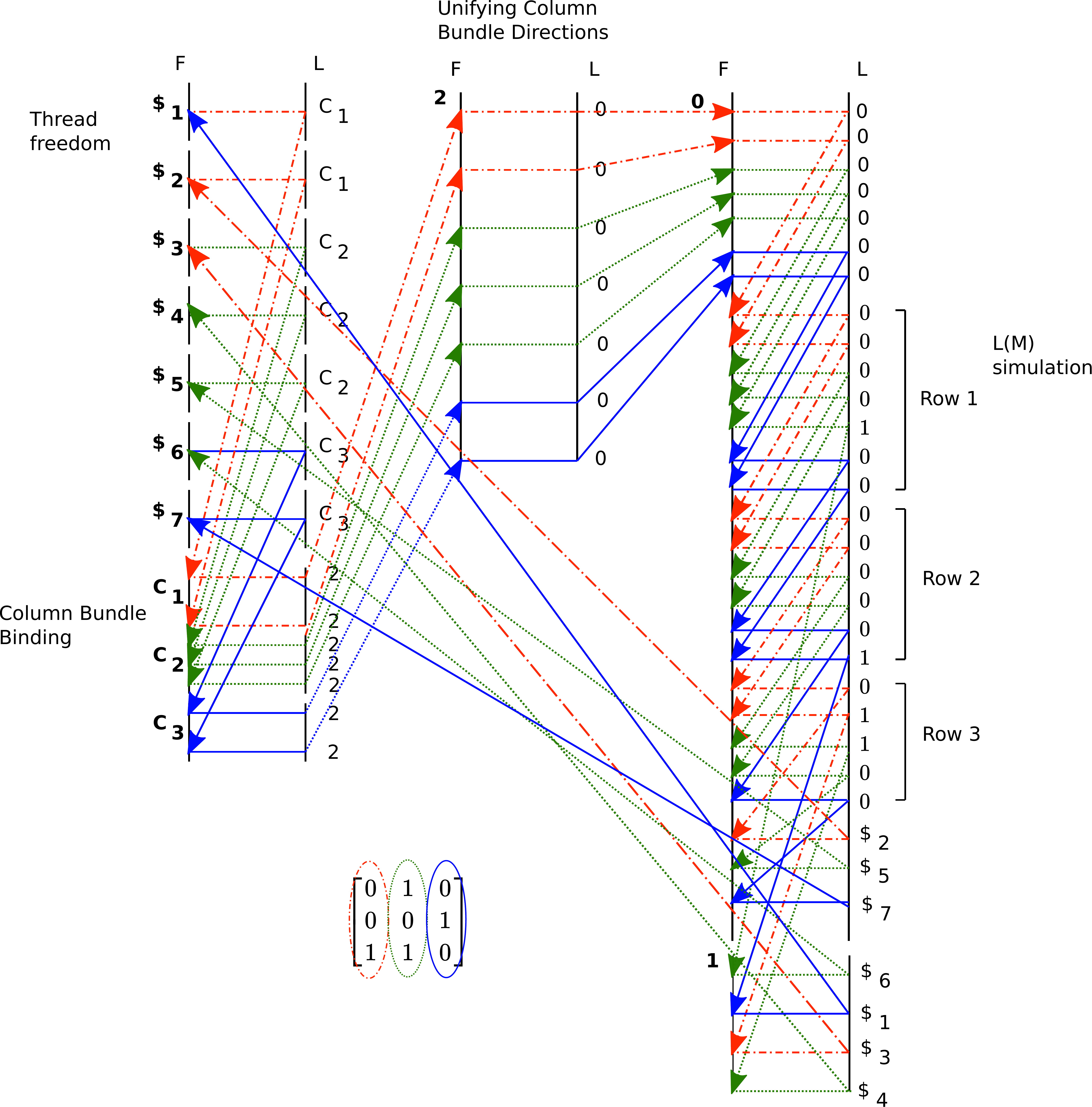}
    \caption{The full simulation of the matrix $M$ is shown. The string resulting from the reduction is $10^32C_3\$_70^52C_3\$_610^22C_2\$_50^52C_2\$_410^42C_2\$_310^42C_1\$_20^52C_1\$_1$. Column bundles are indicated by color and line style. The roles of each part of the reduction are labeled above. `Threads' are allowed to move freely within the column bundles thanks to $\$_i$ characters. The $C_i$ characters bind these bundles together. The 2 symbol unifies the direction of the column bundles. The simulation of $L(M_\pi)$ where $\pi$ is the ordering given by $C_i$ symbols takes place in the 0-block.}
    \label{fig:red_values_full}.
\end{figure}
\textbf{Lining up runs of 1's:} At the point of a 1's appearance in the simulation of $L(M_\pi)$, the corresponding thread is still adjacent to other threads with $0$'s within its column bundle. This highlights a key property of our reduction. \emph{By using the added columns $c_s$ and $c_t$ we ensure that in any solution that we have to carefully consider there are only runs of at most two 1's in $L(M_\pi)$}. This fact allows for our column bundle technique to place these 1's adjacently in the simulation. Moreover, by Property 4 we are ensured that when it is possible the 1's are placed adjacently. Indeed, the role of the $\$_i$ symbols is to allow the threads to move freely within each bundle making Property 4 possible.

\newpage
\textbf{Unifying Column Directions:} What remains is to ensure that the collection of column bundles interacts properly with the $0$-block. This is done using the 2-block which guarantees the column bundles reach the 0-block from the same direction.  See Figure \ref{fig:red_values_full}. The `$- 1$' term arises since by Property 3 the start of the simulation is not counted as a run.


%

The fact that $\rho(BWT(T)) = r_0 + \sigma - 1$ follows from Properties 1-3 which cause every symbol except `1' to contribute exactly one run to $\rho(BWT(T))$ outside of the simulation (1's first appearance is within the simulation).
\end{proof}

\begin{lemma}
\label{lem:AO_desired_L-red}
If all Properties 1-5 hold, then the L-reduction conditions are satisfied.
\end{lemma}

\begin{proof}
By Lemma's \ref{lem:AO_no_desired_suboptimal} and \ref{lem:desired_cost} we have the optimal cost for AO being $r_0^* + \sigma - 1$ and optimal cost for CO as $r_0^* + 1$. For Condition (i) note that $\sigma = \Theta(n)$ and $m+2\ell \leq r_0 \leq 5(m+2\ell)$, so that $r_0^* = \Theta(n)$ ($n$ still referring to the number of vertices in $G$). Hence, we can find an $\alpha$ such that $r_0^* + \sigma - 1 \leq \alpha (r_0^* + 1)$. For Condition (ii) with $\beta = 1$ we have $r_0 + 1 - (r_0^* + 1) \leq r_0+\sigma - 1 - (r_0^* + \sigma - 1)$.
\end{proof}

\begin{lemma}
If any of Properties 1-5 do not hold, then also the L-reduction conditions are satisfied.
\end{lemma}

\begin{proof}
Condition (i) is satisfied since optimal values for the overall problem are unchanged.  For Condition (ii),   
if any of the first four properties are violated, we have already shown in Lemma \ref{lem:AO_desired_L-red} that the inequality holds in the harder case where $\rho(L_{M_\pi})$ has the same value but all of the desired properties are satisfied. If the first four properties hold and the fifth property does not hold, there are two cases. In the first case, if $C_t$ is ordered first and $C_s$ last, then swapping $C_s$ and $C_t$ modifies both sides of the inequality for Condition (ii) by the same amount. In the second case, if either $C_s$ or $C_t$ are not ordered first or last, the simulation can only perform the same or worst than in the corresponding solution for CO, as now there may exist runs of length three which can not be simulated. 
\end{proof}

We have shown an L-reduction from (1,2)-TSP to AO
(with CO as an intermediate problem).
This combined with Lemma~\ref{1,2-TSP-lemma} completes 
 the proof for Theorem \ref{thm:AO_NPC}.\\

\noindent
\textbf{Proof of Corollary \ref{cor:AO_exp_time}:}
Assuming ETH,
there exists no $2^{o(n)}$ time algorithm for Hamiltonian Path Problem~\cite{cygan2015lower}. Our reduction allows us to determine the minimum number of paths in $G$ needed to cover all the vertices and can hence solve Hamiltonian Path. Since the alphabet size $\sigma$ is linear in $n$, an $2^{o(\sigma)}$ time algorithm for AO would imply an $2^{o(n)}$ time algorithm for Hamiltonian Path. At the same time $|T| = \Theta(n^2)$ so that a $2^{o(\sqrt{|T|})}$ time algorithm for AO would again imply an $2^{o(n)}$ time algorithm for Hamiltonian Path. Therefore an algorithm for AO running in time $2^{o(\sigma + \sqrt{|T|})}$ would contradict ETH.

\subsection{Source Ordering on Wheeler Graphs} 
\label{subsec_SO}
We present a brief overview here and details are deferred to Appendix \ref{appendix:wheeler}. To \emph{define the BWT transform from a Wheeler graph $G$ to a string, $BWT(G)$}, we assume a proper ordering on the vertices. We label each vertex in $G$ by its departing edge labels. If a vertex has multiple edge labels leaving it, we consider all possible orderings of its labels and take the one which gives the minimal number of runs. In the proof of Lemma \ref{lem:desired_cost}, by replacing the $C_i$ symbols with sources and constructing the same paths leaving each source as the paths leaving $C_i$ we can obtain Theorem \ref{thm:SO_NPC}. An optimal ordering on the sources provides a minimum ordering of the columns of $M$.

\newpage
\section{Constrained Alphabet Ordering}
\label{sec:DO}

\begin{wrapfigure}[]{r}{0.45\textwidth}
\includegraphics[width = 0.45\textwidth]{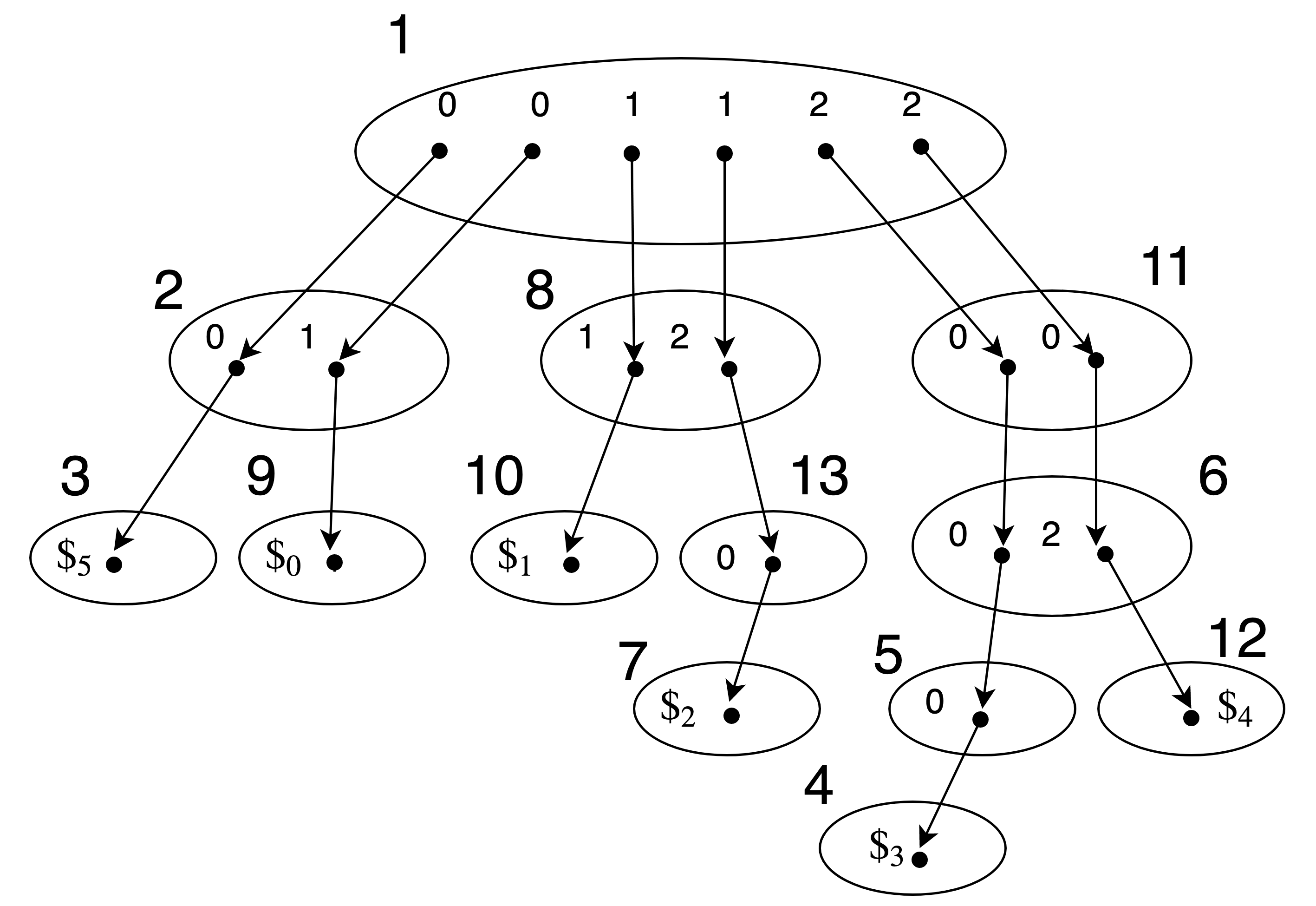}
\caption{The tree for the sting $BWT(T)$ where
$T = 00\$_0 10 \$_1 11 \$_2 021 \$_3 0002 \$_4 202 \$_5$. The order of each block is shown in larger font.}
\label{fig:block_tree}
\end{wrapfigure}

Recall that we wish to find an ordering on the special symbols $\$_0, \hdots, \$_{d-1}$ such that the number of runs in the BWT of $T = T_0\$_0 \hdots T_{d-1}\$_{d-1}$ is minimized. We transform this problem into the problem of ordering a set of paths. This is done by creating a directed path for each substring, $T_i$, and placing these paths in a tree structure where the vertices are grouped into `blocks' which are determined by the labels on their incoming path. The CAO problem then becomes equivalent to finding the starting position of these paths in the root block of the tree. See Figure \ref{fig:block_tree} for an example of such a tree. We wish for an ordering of the paths to minimize the number of runs within blocks, as well as the number of runs between adjacent blocks.
We formally define these ideas. 


For the string $T_i = t_1 t_2 \hdots t_n$ we consider the directed path $P_i$ with vertices labeled 
from beginning to end with symbols $t_n$ $\rightarrow t_{n-1}$ $\hdots$ $ \rightarrow t_1$ $\rightarrow \$_{i-1}$, where $\$_{i-1} = \$_{d-1}$ when $i =0$. For a vertex $v$ in $P_i$, let $str(v)$ denote the string formed by concatenating labels on path $P_i$ from beginning up to, but excluding $v$ (the empty string is possible). 
The block with label $s$, or $B_s$ is defined as $B_s = \{v  \in V: str(v) = s \}$. We consider a block $B_s$ as having a block child $B_{sa}$ which consists of all the vertex being mapped to $B_{sa}$ from $B_s$ with the label $a$. The root of the tree is the block with the empty string as a label. The blocks are ordered by the lexicographic order of the reverse of their strings. 

The blocks of $T = T_0\$_0 \hdots T_{d-1}\$_{d-1}$ can be determined in linear time. This can be done by (i) choosing an arbitrary order on the special symbols, (ii) constructing the BWT while maintaining the original text position in $T$ for each entry, and (iv) constructing the longest common extension structure for $T$. A longest common extension query takes as arguments two indices $i$ and $j$ and returns the length of the longest common sub-string of $T$ starting at $i$ and $j$. The data structure can be constructed in linear time and performs queries in constant time~\cite{DBLP:journals/jda/IlieNT10}. In a linear scan of $BWT(T)$, we can identify where blocks begin using the longest common extension structure. 



Notice that within each block we may permute the ordering of the vertices so that vertices with the same label are consecutive within a block with no effect on the ordering of vertices in the rest of the tree. Therefore, in finding the optimal ordering, we may view each block as a ``tuple", each holding a subset of alphabet symbols appearing only once per tuple. The order in which the tuples are listed is determined by the ordering on the blocks. For example, in Figure \ref{fig:block_tree} the resulting tuples are $(0,1,2)(0,1)(\$_5)(\$_3) (0)(0,2)(\$_2)(1,2)(\$_0)(\$_1)(0)(\$_4)(0)$.  We define a new problem:

\begin{problem}[Tuple Ordering (TO)]
Given a list of tuples $t_1, \hdots, t_q$ in a fixed order, each containing a subset of symbols from $\Sigma$, order the elements in each tuple such the total number of runs in the string formed by their concatenation $t_1 \cdot t_2 \cdot \hdots \cdot t_q$ is minimized (not considering `(', `)' and commas, of course).
\end{problem}



Clearly, the problem of minimizing the total number of runs in $BWT(T)$ is equivalent to maximizing the total number of adjacent matches between the tuples formed as above.
We begin by first sorting the values within each tuple. This can easily be done in linear time if the  values within each tuple are bounded. For each tuple $t_i$ we will create two lists, $L_i$ and $R_i$, both of which contain the elements of that tuple in sorted order. 
\begin{itemize}
    \item We start by marking all elements in $L_1$ as available.
    \item Next, for $i \geq 1$, assume inductively that at least one element in $L_i$ is marked as available.
    \item If $|L_i| \geq 2$:
    \begin{itemize}
        \item If two elements in $L_i$ are marked then we mark all elements in $R_i$ as available.
        \item Else if only one element in $L_i$ is available, we mark all but the matching element in $R_i$ as available.
    \end{itemize}
    \item Else if $|L_i| = 1$, mark the one element in $R_i$ as available too.
    \item To mark the availability of elements in $L_{i+1}$ given the availability of elements in $R_i$ we mark an element in $L_{i+1}$ iff its matching value in $R_i$ is marked. If this results in none of the $L_{i+1}$ being marked then we
    make all $L_{i+1}$ marked as available.
\end{itemize}

Once marking is completed for each tuple $t_i$, we take some marked element from $L_i$  as the left most element of that tuple, and some marked element marked in $R_i$ as the rightmost element of that tuple. A simple exchange argument shows that this greedy algorithm returns an optimal ordering of the tuples to maximize adjacent matches.

\textbf{Time Complexity of CAO Algorithm:} 
Creating the blocks/tuples can be done in linear time.
The total number of elements across all tuples is $O(N)$, where $N$ is $|T_0| + \hdots |T_{d-1}|$. Moreover, the algorithm for finding an optimal ordering of elements in the tuples runs in time proportional to the total number of elements, which is  $O(N)$.

\textbf{An Example:}
We will now show an example where the special symbol ordering greatly reduces the number of runs in the BWT. Let $d$ be the number of strings and $n$ the length of the strings. It is possible for a set of special symbols to be ordered such that the number of runs is $\Omega(nd)$. Let $n = \log_\sigma d$ and consider the $\sigma^n$ distinct binary strings concatenated with special symbols in lexicographic order. For example, with $n = 3$ we would have $T = 000\$_0 001\$_1 010 \$_2 011 \$_3 100 \$_4 101 \$_5 110 \$_6 111 \$_7$ with $\$_0 < \$_1 \hdots < \$_7$. The string $BWT(T)$ alternates between the $\$$'s, 0's, and 1's yielding $\Omega(nd)$ runs.
On the other hand, for this same case, arranging the $\$$'s in the optimal ordering allows for at most two runs per block giving $O(d)$ runs in total.

\section{Discussion and Open Problems}
\label{sec:open_problems}
Apart from the obvious open question of whether we can develop a better approximation for AO (from a lower bounds and upper bounds perspective), the following problem, a generalization of the Constrained Alphabet Ordering has unknown computational complexity. 

\begin{problem}[Constrained Alphabet Ordering with arbitrary placement]
 Given a string $T$ and a set of symbols $\$_i$ $(0\leq i \leq d-1)$ each occurring at most once in $T$ find a ordering $\pi$ on the symbols $\$_i$ such that $\pi$ is compatible with the natural ordering given on ${0, \hdots, \sigma-1}$, and the number of runs in $BWT(T_0\$_0T_1\$_1, \hdots T_{d-1}\$_{d-1})$ is minimized.
\end{problem}

Like CAO, the problem seems to be of more significance when the number of strings is far more than the length of strings $n$. 
The constraints placed on the symbol ordering in this problem lie somewhere between the constraints in AO and CAO. 
Likewise, for a given set of strings the optimal number of runs for CAO with arbitrary placement lies somewhere between the optimal number under AO and the optimal number under CAO.
We showed here AO is NP-hard while CAO is solvable in polynomial time. The key element in the reductions used for AO is having multiple strings containing the same symbol $a$, where $a$'s order needs to be determined in the alphabet. This property no longer holds in CAO with arbitrary placement, but still no polynomial time algorithm is evident, making the question of this problem's computational complexity intriguing. 

\bibliography{references}
\section{Appendix}
\appendix


\section{Wheeler Graphs}
\label{appendix:wheeler}
Recently, a new line of research into a class of graphs called Wheeler graphs
with properties related to the LF-mapping has been initiated~\cite{DBLP:conf/soda/AlankoDPP20,DBLP:journals/tcs/GagieMS17,GibneyT19}. 
These graphs can easily model the BWT, de-Brujin graphs, variation graphs, etc. Wheeler graphs admit a space efficient index which allows for matching patterns in optimal time. 
Let $(u,v,k)$ denote the directed edge from $u$ to $v$ with label $k$.
A Wheeler graph is a directed graph with edge labels where there exists an ordering $\phi$ on the vertices such that for any two edges $(u,v,k)$ and $(u',v',k')$: (i) $k < k' \implies v <_\phi v'$ and (ii) $(k = k')\land(u <_\phi u') \implies v \leq_\phi v'$ and 
and vertices with in-degree zero must be placed first in the ordering. We consider an ordering which satisfies Conditions (i) and (ii) to be a \emph{proper ordering}.

When the vertices are in a proper ordering, if we lay the vertices of a Wheeler graph $G$ out as though they were in the LF-mapping (see Figure \ref{fig:LF_map}) we would see that no edges with the same label cross. Also, we would see that all edges with the same label get mapped into the same portion of $F$ in the LF-mapping. The differences between Wheeler graphs and previous case lies in the fact that now vertices can map edges to multiple places simultaneously. 

\begin{figure}[h]
    \centering
    \includegraphics[width=0.8\textwidth]{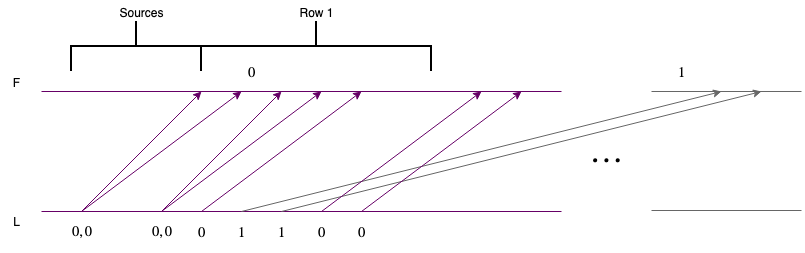}
    \caption{Reduction from CO\tdagger\ to SO. We use sources to bind together paths needed for each column.}
    \label{fig:SO}
\end{figure}

Using the natural definition of $BWT(G)$ presented in Section \ref{sec:AO}, the more general properties of Wheeler graphs makes the reduction from CO easier. Instead of using extra symbols to bind together paths induced by the same column of modified incidence matrix, we simply bind them together at the same source. 
To be more precise, given the modified matrix $M$ the input to the problem CO, construct a graph as follows:
For $1\leq j \leq n$, $0 \leq i \leq m + 2\ell $, if $i = 0$ simply create a new source vertex $s_j$, 
and $M_{i,j} = 1$ we construct a directed path starting at $s_j$ with edge labels $0^{i+1}1$, and if $i = m + 2\ell+1$ we construct the directed path with labels $0^{i+1}$ again rooted at $s_j$. See Figure \ref{fig:SO} for an illustration.
Note that this graph is a forest and hence a Wheeler graph.


\end{document}